\documentclass[conference,twocolumn]{IEEEtran}
\usepackage{geometry}
\usepackage[english]{babel}
\usepackage{enumitem}
\usepackage[latin1]{inputenc}
\usepackage{enumerate}
\usepackage{color}
\usepackage[T1]{fontenc}
\usepackage{subfigure}
\usepackage{dsfont}
\usepackage[final]{graphicx}
\usepackage[T1]{fontenc}
\usepackage{amsmath}
\usepackage{mathtools, cuted}
\usepackage{amsthm}
\usepackage{amstext}
\usepackage{amssymb}
\usepackage{mathrsfs}
\usepackage{cite}
\usepackage{mathtools}
\usepackage{tikz}
\usepackage{marginnote}
\usepackage{tkz-tab}
\usepackage{pgfplots}
\usepackage{stackengine}

\usetikzlibrary{arrows,positioning, shapes}
\usepackage{float}

\usepackage{setspace}
\singlespace

\usepackage{hyperref}
\def\R{\mathbb{R}}

\def\eps{\varepsilon}

\def\D{{\mathcal D}}
\def\E{{\mathbb E}}
\def\F{{\mathcal F}}
\def\H{{\mathcal H}}

\def\P{{\mathcal P}}

\def\X{{\mathcal X}}
\def\Y{{\mathcal Y}}
\def\M{{\mathcal M}}

\def\Z{{\mathcal Z}}
\def\U{{\mathcal U}}
\def\V{{\mathcal V}}

\def\sBer{{\mathsf{Bernoulli}}}
\def\sP  {{\mathsf{P}}}

\def \cP {\mathsf{P}_{\mathsf{c}}}

\newcommand{\repdc}[3]{#1_{#2} , \ldots , #1_{#3}}
\usepackage[font=small,labelsep=space]{caption}
\DeclareCaptionLabelSeparator{dot}{.~}
\newcounter{example}

\newtheorem{definition}{Definition}
\newtheorem{theorem}{Theorem}
\newtheorem{corollary}{Corollary}

\newtheorem{proposition}{Proposition}
\newtheorem{lemma}{Lemma}
\theoremstyle{remark}

\newcommand{\markov}{\mathrel\multimap\joinrel\mathrel-%
\mspace{-9mu}\joinrel\mathrel-}

\usepgfplotslibrary{fillbetween}
\usepackage{times}
\usepackage{tikz}
\usepackage{amsmath}
\usepackage{verbatim}
\usetikzlibrary{arrows,shapes}
\tikzstyle{RectObject}=[rectangle,fill=white,draw,line width=0.2mm]
\tikzstyle{line}=[draw]
\tikzstyle{arrow}=[draw, -latex]
\usetikzlibrary{decorations.pathmorphing}
\usetikzlibrary{calc,shapes, positioning}
\usepackage{graphicx}
\usepackage{caption}
\usetikzlibrary{shapes.geometric}

\DeclareFontFamily{U}{BOONDOX-calo}{\skewchar\font=45 }
\DeclareFontShape{U}{BOONDOX-calo}{m}{n}{
	<-> s*[1.05] BOONDOX-r-calo}{}
\DeclareFontShape{U}{BOONDOX-calo}{b}{n}{
	<-> s*[1.05] BOONDOX-b-calo}{}
\DeclareMathAlphabet{\mathcalboondox}{U}{BOONDOX-calo}{m}{n}
\SetMathAlphabet{\mathcalboondox}{bold}{U}{BOONDOX-calo}{b}{n}
\DeclareMathAlphabet{\mathbcalboondox}{U}{BOONDOX-calo}{b}{n}

\IEEEoverridecommandlockouts

\allowdisplaybreaks

\begin{document}

\title{\vspace{5.5mm}{Privacy-Aware Guessing Efficiency}}
\author{\IEEEauthorblockN{Shahab Asoodeh\thanks{This work was supported in part by NSERC of Canada.}, Mario Diaz, Fady Alajaji, and Tam\'{a}s Linder\thanks{The authors are with the Department of Mathematics and Statistics, Queen's University, Canada. Emails: \{asoodehshahab, fady, linder\}@mast.queensu.ca, 13madt@queensu.ca.}}}
\maketitle

\begin{abstract}
We investigate the problem of guessing a discrete random variable $Y$ under a privacy constraint dictated by another correlated discrete random variable $X$, where both guessing efficiency and privacy are assessed in terms of the probability of correct guessing. We define $\mathcalboondox{h}(P_{XY}, \eps)$ as the maximum probability of correctly guessing $Y$ given an auxiliary random variable $Z$, where the maximization is taken over all $P_{Z|Y}$ ensuring that the probability of correctly guessing $X$ given $Z$ does not exceed $\eps$. We show that the map $\eps\mapsto \mathcalboondox{h}(P_{XY}, \eps)$ is strictly increasing, concave, and piecewise linear, which allows us to derive a closed form expression for $\mathcalboondox{h}(P_{XY}, \eps)$ when $X$ and $Y$ are connected via a binary-input binary-output channel. For $\{(X_i, Y_i)\}_{i=1}^n$ being pairs of independent and identically distributed binary random vectors, we similarly define $\underbar{\emph{h}}_n(P_{X^nY^n}, \eps)$ under the assumption that $Z^n$ is also a binary vector. Then we obtain a closed form expression  for $\underbar{\emph{h}}_n(P_{X^nY^n}, \eps)$ for sufficiently large, but nontrivial values of $\eps$.
\end{abstract}

\section{Introduction and Preliminaries}
Given private information, represented by a random variable $X$, non-private observable information, say $Y$, is generated via a fixed channel $P_{Y|X}$. Consider two communicating agents Alice and Bob, where Alice observes $Y$ and wishes to disclose it to Bob as accurately as possible in order to receive a payoff, but in such a way that $X$ is kept almost private from him.  Given the joint distribution $P_{XY}$, Alice chooses a random mapping $P_{Z|Y}$, a so-called privacy filter, to generate a new random variable $Z$, called the \emph{displayed data}, such that Bob can \emph{guess} $Y$ from $Z$ with as small error probability as possible while $Z$ cannot be used to efficiently guess $X$.

The tradeoff between utility and privacy was addressed from an information-theoretic viewpoint in \cite{yamamotoequivocationdistortion,Asoode_submitted,Asoodeh_Allerton,Calmon_fundamental-Limit,Funnel}, where both utility and privacy were measured in terms of information-theoretic quantities. In particular, in \cite{Asoode_submitted} both utility and privacy were measured in terms of the mutual information $I$. Specifically, the so-called \emph{rate-privacy function} $g(P_{XY}, \eps)$ was defined as the maximum of $I(Y;Z)$ over all $P_{Z|Y}$ such that $I(X;Z)\leq \eps$. In the most stringent privacy setting $\eps=0$, called \emph{perfect privacy}, it was shown that $g(P_{XY}, 0)>0$ if and only if $X$ is weakly independent of $Y$, that is, if the set of
vectors $\{P_{X|Y}(\cdot|y):y\in \Y\}$ is linearly dependent. In \cite{Calmon_fundamental-Limit}, an equivalent result was obtained  in terms of the singular values of the operator $f\mapsto \E[f(X)|Y]$.
Although a connection between this information-theoretic privacy measure and a coding theorem is established in \cite{Asoode_submitted} and \cite{El_Gamal_Gray_Wyner}, the use of mutual information as a privacy measure is not satisfactorily motivated in an \emph{operational} sense. To find a measure of privacy with a clear operational meaning, in this paper we take an estimation-theoretic approach and define both privacy and utility measures in terms of the probability of guessing correctly. 

Given discrete random variables $U\in \U$ and $V\in \V$, the probability of correctly guessing $U$ given $V$ is defined as
\begin{equation*}
	 \cP(U|V)\coloneqq\max_{g}\Pr(U=g(V))=\sum_{v\in \V}\max_{u\in \U}\ P_{UV}(u, v),
	\end{equation*}
where the first maximum is taken over all functions ${g:\V~\to \U}$.
It is easy to show that $\cP$ satisfies the data processing inequality, i.e., $\cP(U|W)\leq \cP(U|V)$ for $U$, $V$ and $W$ which form the Markov chain $U\markov V\markov W$. Thus, we measure privacy in terms of $\cP(X|Z)$ which quantifies the advantage of an adversary observing $Z$ in guessing $X$ in a single shot attempt.

A similar operational measure of privacy was recently proposed in \cite{Issa_Sibson}, where $P_{Z|X}$ is said to be $\eps$-private if $\log\frac{\cP(U|Z)}{\cP(U)}\leq \eps$ for \emph{all} auxiliary random variables $U$ satisfying $U\markov X\markov Z$. This requirement guarantees that no \emph{randomized} function of $X$ can be efficiently estimated from $Z$, which leads to a strong privacy guarantee. In \cite{Fawaz_Makhdoumi}, maximal correlation \cite{gebelien} was proposed as another measure of privacy.
Operational interpretations corresponding to this privacy measure are given in  \cite{Calmon_bounds_Inference} for the discrete case and  in \cite{Asoode_MMSE_submitted} for a continuous setup.

To quantify the conflict between utility and privacy, we define the \emph{privacy-aware guessing function} $\mathcalboondox{h}$ as
\begin{equation}\label{Def_h_eps}
  \mathcalboondox{h}(P_{XY}, \eps)\coloneqq\sup_{P_{Z|Y}:X\markov Y\markov Z, \atop \cP(X|Z)\leq \eps}\cP(Y|Z).
\end{equation}
Due to the data processing inequality, we can restrict the privacy threshold $\eps$ to the interval $[\cP(X), \cP(X|Y)]$, where $\cP(X)$ is the probability of correctly guessing $X$ in the absence of any side information. For $\eps$ close to $\cP(X)$, the privacy guarantee $\cP(X|Z)\leq \eps$ intuitively means that it is nearly as hard to guess $X$ observing $Z$ as it is without observing $Z$.

We derive functional properties of the map $\eps\mapsto \mathcalboondox{h}(P_{XY}, \eps)$. In particular, we show that it is strictly increasing, concave, and piecewise linear. Piecewise linearity (Theorem~\ref{Thm:PiecewiseLinearity}), which is the most important and technically difficult result in the paper, allows us to derive a tight upper bound on $\mathcalboondox{h}(P_{XY}, \eps)$ for general $P_{XY}$. As a consequence of concavity, we derive a closed form expression for $\mathcalboondox{h}(P_{XY}, \eps)$ for any $\eps\in [\cP(X), \cP(X|Y)]$ when $X$ and $Y$ are both binary. It is shown (Theorem~\ref{Theorem_Linearity_BIBO}) that either the Z-channel or the \textit{reverse} Z-channel achieves $\mathcalboondox{h}(P_{XY}, \eps)$ in this case depending on the backward channel.


We also consider the vector case for a pair of binary random vectors $(X^n, Y^n)$ under an additional constraint that $Z^n$ is a binary random vector. Here,  $Z^n$ is revealed publicly and the goal is to guess $Y^n$ under the privacy constraint $\cP(X^n|Z^n)\leq \eps^n$. This model can be viewed as a privacy-constrained version of the \emph{correlation distillation} problem studied in \cite{Witsenhausen:dependent}. Suppose Alice  and Bob respectively observe $Y^n$ and $Z^n$, where $\{(Y_i, Z_i)\}_{i=1}^n$ is independent and identically distributed (i.i.d.) according to the joint distribution $P_{YZ}$, and assume that they are to design non-constant Boolean functions $f$ and $g$ such that $\Pr(f(Y^n)=g(Z^n))$ is maximized. A dimension-free upper bound for this probability was given in \cite{Witsenhausen:dependent}. Now suppose $P_{YZ}$ is not given and Alice is to design $P_{Z|Y}$ (for a fixed $\Y$-marginal) that maximizes $\cP(f(Y^n)|Z^n)$ for a given function $f$ while $\cP(X^n|Z^n)\leq \eps^n$. We show (Theorem~\ref{Thm:Difference}) that if $\{(X_i, Y_i)\}_{i=1}^n$ is i.i.d. according to $P_{XY}$ with $|\X|=|\Y|=2$ and $P_{Y|X}$ is a binary symmetric channel, then the maximum of $\cP(Y^n|Z^n)$ under the privacy constraint $\cP(X^n|Z^n)\leq \eps^n$ admits a closed form expression for sufficiently large but nontrivial $\eps$. This then provides a lower bound for the privacy-constrained correlation distillation problem due to the trivial fact that $\cP(f(Y^n)|Z^n)\geq \cP(Y^n|Z^n)$ for any function $f$.



We omit the proof of most of the results due to space limitations. The proofs are available in \cite{Asoode_H_EPSILON}.

\section{Scalar Case}
Suppose $X$ and $Y$ are discrete random variables with finite alphabets $\X=\{1, \dots, M\}$ and $\Y=\{1, \dots, N\}$, res\-pectively,  and with joint distribution $\mathsf{P}=\{P_{XY}(x,y), x\in \X, y\in \Y\}$, whose marginals over $\X$ and $\Y$ are $(p_1, \dots, p_{M})$ and $(q_1, \dots, q_{N})$, respectively. Let $X$ re\-present the private data and $Y$ represent a non-private measurement of $X$, which, upon passing it via a privacy filter $P_{Z|Y}$, is publicly displayed as $Z$. In order to quantify the conflict between privacy with respect to $X$ and utility with respect to $Y$, the so-called rate-privacy function $g(\sP, \eps)$ was introduced in  \cite{Asoode_submitted}. In what follows, we use Arimoto's mutual information to generalize this definition.
\subsection{The Utility-Privacy Function of Order $(\nu, \mu)$}
Let $H_\nu(X)$ and $H_\nu^\mathsf{A}(X|Z)$ denote respectively the R\'{e}nyi entropy of order $\nu$ and Arimoto's conditional entropy of order $\nu$ \cite{Verdu_ALPHA}, defined for $\nu>1$ as
$$H_\nu(X)\coloneqq\frac{1}{1-\nu}\log\left(\sum_{x\in \X}P_X^\nu(x)\right),$$ and
$$H_\nu^\mathsf{A}(X|Z)\coloneqq\frac{\nu}{1-\nu}\log\left(\sum_{z\in \Z}\left[\sum_{x\in \X}P_{XZ}^\nu(x, z)\right]^{1/\nu}\right).$$ We define (by continuity) $H_1(X)=H(X)$, $H_1^\mathsf{A}(X|Z)=H(X|Z)$,  $H_\infty(X)=-\log \cP(X)$, and $H_\infty^\mathsf{A}(X|Z)=-\log\cP(X|Z).$ \emph{Arimoto's mutual information of order $\nu\geq 1$} is defined as (see, e.g., \cite{Verdu_ALPHA}) $$I^\mathsf{A}_\nu(X;Z)\coloneqq H_\nu(X)-H_\nu^\mathsf{A}(X|Z).$$ Thus $I^\mathsf{A}_1(X;Z)=I(X;Z)$.
\begin{definition}\label{Def_RPF}
  For a given joint distribution $\mathsf{P}$ and a pair $(\nu, \mu)$, $\nu, \mu\in[1, \infty]$, the utility-privacy function of order $(\nu, \mu)$ is
  $$g^{(\nu, \mu)}(\mathsf{P}, \eps)\coloneqq\max_{P_{Z|Y}\in \mathcalboondox{D}^{\nu}(\sP, \eps)}I^\mathsf{A}_\mu(Y;Z),$$
where $$\mathcalboondox{D}^\nu(\mathsf{P}, \eps)\coloneqq\{P_{Z|Y}: X\markov Y\markov Z, I^\mathsf{A}_\nu(X;Z)\leq \eps\}.$$
\end{definition}
Note that $\mathcalboondox D^\nu(\mathsf{P}, \eps)$ cannot be empty since all channels $P_{Z|Y}$ with $Z$ independent of $X$ satisfy $I_{\nu}^\mathsf{A}(X;Z)=0$, and so they belong to $\D^\nu(\mathsf{P}, \eps)$ for any $\eps\geq 0$. Using a similar technique as in \cite{Witsenhausen_Wyner}, one can show that $\eps\mapsto g^{(\nu, \mu)}(\sP, \eps)$ is strictly increasing for any $\nu, \mu\geq 1$. It is also worth mentioning that an application of Minkowski's inequality implies that the map $P_{Z|Y}\mapsto \exp\left\{\frac{(\nu-1)}{\nu}I^\mathsf{A}_\nu(Y;Z)\right\}$ is convex for $\nu\geq 1$, and thus the maximum in the definition of $g^{(\nu, \mu)}(\mathsf{P}, \eps)$ is achieved at the boundary of the feasible set where $I_\nu^\mathsf{A}(X;Z)=\eps$.
We denote $g^{(\infty, \infty)}(\mathsf{P}, \eps)$ and $g^{(1,1)}(\mathsf{P}, \eps)$ respectively by $g^{\infty}(\mathsf{P}, \eps)$ and $g(\mathsf{P}, \eps)$. Since $I_\infty(Y;Z)=\log\frac{\cP(Y|Z)}{\cP(Y)}$, 
$g^\infty(\mathsf{P}, \eps)$ can be equivalently described as the smallest $\Gamma\geq 0$ such that
$\cP(Y|Z)\leq \mathsf{P}_{\mathsf{c}}(Y) 2^\Gamma,$
for every $P_{Z|Y}$ satis\-fying
$\cP(X|Z)\leq \cP(X) 2^\eps.$
We note that for small $\eps$ the condition $I^\mathsf{A}_\infty(X;Z)\leq \eps$ intuitively means that it is nearly as hard for an \emph{adversary} observing $Z$ to predict $X$ as it is without $Z$. Therefore, $g^\infty(\sP, 0)$ quantifies the efficiency of guessing $Y$ from $Z$ such that $\cP(X|Z)=\cP(X)$. It is thus interesting to obtain a necessary and sufficient condition for $\sP$ under which $g^\infty(\sP, 0)>0$. We obtain such a condition for the special case of binary $X$ and $Y$ in the next section.

In general, the map $\nu\mapsto I_\nu^\mathsf{A}(X;Z)$ is not monotonic\footnote{It is relatively easy to show that if $X$ is uniformly distributed, then $I_\nu^\mathsf{A}(X;Z)$ coincides with Sibson's mutual information  of order $\nu$ \cite{Verdu_ALPHA} which is known to be increasing in $\nu$ \cite[Theorem 4]{Verdu_Convexity}. Consequently,  $\nu\mapsto I_\nu^\mathsf{A}(X;Z)$ is increasing over $(1, \infty]$ if $X$ is uniformly distributed.} and hence $P_{Z|Y}$ might belong to $\mathcalboondox D^\nu(\mathsf{P}, \eps)$ but not to $\mathcalboondox D^\mu(\mathsf{P}, \eps)$ for $\mu<\nu$. Nevertheless, the following lemma allows us to obtain upper and lower bounds for $g^{(\nu, \mu)}(\mathsf{P}, \cdot)$ in terms of $g^{\infty}(\mathsf{P}, \cdot)$.
\begin{lemma}\label{Lemma_Bounds_g_Infy}
Let $(X,Y)$ be a pair of random variables having joint distribution $\sP$ and $\nu,\mu\in(1,\infty)$. Then
\begin{equation*}
g^{(\nu,\mu)}(\sP, \eps) \leq g^\infty(\sP, \psi(\nu, \eps))+H_\mu(Y)-H_\infty(Y),
\end{equation*}
where $\psi(\nu, \eps)\coloneqq\frac{\nu-1}{\nu}\eps + \frac{1}{\nu}H_\infty(X)$. Furthermore, we have for $\eps\geq H_\nu(X)-H_\infty(X)$  that
\begin{equation*}
g^{(\nu,\mu)}(\sP,\eps) \geq \frac{\mu}{\mu-1}g^\infty(\sP,\varphi(\nu, \eps))-\frac{1}{\mu-1}H_\infty(Y),
\end{equation*}
where $\varphi(\nu, \eps)\coloneqq \eps-H_\nu(X)+H_\infty(X)$.
\end{lemma}

This lemma shows that the family of functions $g^{(\nu, \mu)}(\mathsf{P}, \eps)$ for $\nu, \mu>1$ can be bounded from above and below by $g^\infty(\mathsf{P}, \delta)$, where $\delta$ depends on $\eps$ and $\nu$. The case $\nu=\mu=1$ is studied in \cite{Asoode_submitted}. As a result, in the following section we only focus on $g^\infty(\mathsf{P}, \eps)$.  It turns out that it is easier to study $\mathcalboondox{h}(\mathsf{P}, \eps)$, defined in \eqref{Def_h_eps}, instead. It is straightforward to verify that
$$g^\infty(\mathsf{P}, \eps)=\log\frac{\mathcalboondox{h}(\mathsf{P}, 2^\eps\cP(X))}{\cP(Y)},$$ and hence all the results for $\mathcalboondox{h}(\mathsf{P}, \eps)$ can be translated to results for $g^\infty(\mathsf{P}, \eps)$. In particular, perfect privacy $g^\infty(\mathsf{P}, 0)$ corresponds to $\mathcalboondox{h}(\mathsf{P}, \cP(X))$.
Notice that $\mathcalboondox{h}(\sP, \cP(X))>\cP(Y)$ is equivalent to $g^{\infty}(\sP,0)>0$. As opposed to $I_\nu(X;Z)$ with $1\leq\nu<\infty$, $I_\infty(X;Z)=0$ does not necessarily imply the independence of $X$ and $Z$ (unless $X$ is uniformly distributed). In particular, the  weak independence\footnote{Using a similar proof as in \cite{Asoode_submitted}, it can be shown that $g^{(\nu, \mu)}(\sP, 0)>0$ for $\nu, \mu\in [1, \infty)$ if and only if $X$ is weakly independent of $Y$.} argument from \cite[Lemma 10]{Asoode_submitted} (see also \cite{Calmon_fundamental-Limit}) cannot be applied for $g^\infty$.  For the sake of brevity, we simply write $\mathcalboondox{h}(\eps)$ for $\mathcalboondox{h}(\sP, \eps)$  when there is no risk of confusion.
\subsection{Privacy-Aware Guessing Function}
 It is clear from \eqref{Def_h_eps} that $\cP(Y)\leq \mathcalboondox{h}(\eps)\leq 1$, and $\mathcalboondox{h}(\eps)=1$ if and only if $\eps\geq\cP(X|Y)$. A direct application of the Support Lemma \cite[Lemma 15.4]{csiszarbook} shows that it is enough to consider random variables $Z$ supported on $\Z=\{1,\ldots,N+1\}$. Thus, the privacy filter $P_{Z|Y}$ can be realized by an $N\times (N+1)$ stochastic matrix $F$. Let $\F$ be the set of all such matrices. Then both utility $\U(\sP,F)=\cP(Y|Z)$ and privacy $\P(\sP, F)=\cP(X|Z)$ are functions of $F\in \F$ and we can express $\mathcalboondox{h}(\eps)$ as $$\mathcalboondox{h}(\eps)=\max_{F\in \F, \atop \P(\sP, \eps)\leq \eps}\U(\sP, F).$$
%
It can be verified that $F\mapsto\P(\sP,F)$ and
$F\mapsto\U(\sP,F)$ are continuous convex functions over $\F$.   It can also be shown that the set $$\mathcal{R}\coloneqq \{(\P(\sP, F), \U(\sP, F)):~F\in \F\}$$ is convex. Furthermore, since the graph of  $\mathcalboondox{h}(\eps)$ is the upper boundary of $\mathcal{R}$, we conclude that $\eps\mapsto \mathcalboondox{h}(\eps)$ is concave, and so it is strictly increasing and continuous on $[\cP(X),\cP(X|Y)]$. As a consequence, for every $\eps\in[\cP(X),\cP(X|Y)]$ there exists $G$ such that $\P(\sP,G)=\eps$ and $\U(\sP,G)=\mathcalboondox{h}(\eps)$. We call such a privacy filter $G$ optimal at $\eps$.


The following theorem reveals that $\mathcalboondox{h}(\cdot)$ is a piecewise linear function, as depicted in Fig.~\ref{fig:Typical_h}.

\begin{theorem}
\label{Thm:PiecewiseLinearity}
The function $\mathcalboondox{h}:[\cP(X),\cP(X|Y)]\to\R$ is piecewise linear, i.e., there exist $K\geq1$ and thresholds $\cP(X)=\eps_0<\eps_1<\ldots<\eps_K=\cP(X|Y)$ such that $\mathcalboondox{h}$ is linear on $[\eps_{i-1}, \eps_i]$ for all $1\leq i\leq K$.
\end{theorem}
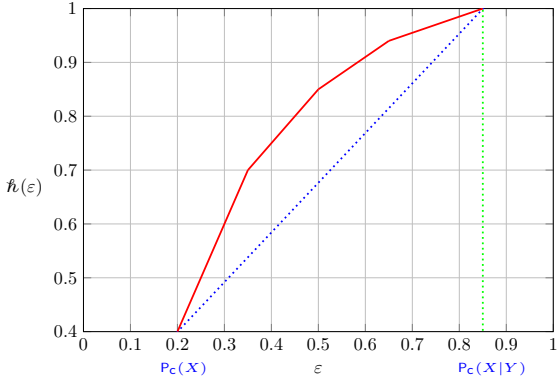
\begin{figure}
\hspace{0.1cm}
\begin{tikzpicture}[thick, scale=0.68]
\begin{axis}[width = 4.2in, height = 2in, xmin=0,xmax=1,ymin=0.4,ymax=1, y=10.5cm,
  samples=100,  y label style={at={(axis description cs:0.01,0.4)},rotate=-90,anchor=south},
  grid=both,xlabel= $\eps$, ylabel=${\mathcalboondox{h}(\eps)}$,
  no markers]
\addplot[no marks,blue,dotted, line width=1pt,samples=600] {0.923*x+0.215};
\addplot[domain=0.2:0.35][ no marks,red,line width=1pt,samples=600] {2*x};
\addplot[domain=0.35:0.5][no marks,red,line width=1pt,samples=600] {x+0.35};
\addplot[domain=0.5:0.65][no marks,red,line width=1pt,samples=600] {0.6*x+0.55};
\addplot[domain=0.65:0.85][no marks,red,line width=1pt,samples=600] {0.3*x+0.745};
\addplot[mark=none, green, dotted, line width=1pt] coordinates {(0.85,0) (0.85,1)};
\end{axis}
\node[below,blue] at (7.91, -0.39) {\tiny{$\cP(X|Y)$}};
\node[below,blue] at (1.95, -0.39) {\tiny{$\mathsf{P}_\mathsf{c}(X)$}};
\end{tikzpicture}
\caption{Typical graph of $\mathcalboondox{h}(\eps)$. The dotted line represents the chord connecting $(p, \mathcalboondox{h}(p))$ and $(\cP(X|Y), 1)$ which can be viewed as a trivial lower bound for $\mathcalboondox{h}(\cdot)$.}
\label{fig:Typical_h}
\end{figure}
Consider the map $\H:\F\to[0,1]\times[0,1]$ given by $\H(\sP, F) = (\P(\sP, F),\U(\sP, F))$. Let $\D \coloneqq \left\{D\in\M_{N \times N+1}: \|D\|=1\right\}$, where $||\cdot||$ denotes the Euclidean norm on $\M_{N\times (N+1)}$, the set of real matrices of size $N\times (N+1)$. For $G\in\F$ define  $$\D(G) \coloneqq \left\{D\in \D: G+tD\in \F \textnormal{ for some } t>0\right\}.$$ The proof of the previous theorem is heavily based on the following technical, yet intuitive, result: for every $G\in\F$, there exists $\delta>0$ such that $\H$ is linear on $[G,G+\delta D]$ for \textit{every} $D\in\D(G)$.

The proof technique allows us to derive the slope of $\mathcalboondox{h}$ on $[\eps_{i-1}, \eps_i]$, given the family of optimal filters at a single point $\eps\in [\eps_{i-1}, \eps_i]$. For example, since the family of optimal filters at $\eps=\cP(X|Y)$ is easily obtainable, it is then possible to compute $\mathcalboondox{h}$ on the last interval. In the binary case, this observation and the concavity of $\mathcalboondox{h}$ allow us to show that $\mathcalboondox{h}$ is linear on its entire domain $[\cP(X), \cP(X|Y)]$.
\vspace{-0.02cm}
\subsection{Binary Case}
Assume now that $X$ and $Y$ are both binary. Let $\mathsf{BIBO}(\alpha,\beta)$ denote a binary input binary output channel from $X$ to $Y$ with  $P_{Y|X}(\cdot|0)=(\bar{\alpha}, \alpha)$ and $P_{Y|X}(\cdot|1)=(\beta, \bar{\beta})$, where $\bar{x}\coloneqq 1-x$ for $x\in [0,1]$.
Notice that if $X\sim\sBer(p)$ with $p\in[\frac{1}{2},1)$, then $\cP(X)=p$ and hence $\mathcalboondox{h}(p)$ corresponds to the maximum of $\cP(Y|Z)$ under perfect privacy $\cP(X|Z)=p$. Furthermore, if $P_{Y|X}=\mathsf{BIBO}(\alpha,\beta)$ with $\alpha,\beta\in[0,\frac{1}{2})$, then we have $\cP(X|Y)=\max\{\bar{\alpha}\bar{p}, \beta p\}+\bar{\beta}p.$ 
Notice that if $\bar{\alpha}\bar{p}\leq\beta p$, then  $\cP(X|Y)=\cP(X)=p$. 

The binary symmetric channel with crossover probability $\alpha$, denoted by $\mathsf{BSC}(\alpha)$, and also the Z-channel with crossover probability $\beta$, denoted by $\mathsf{Z}(\beta)$, are both examples of $\mathsf{BIBO}(\alpha, \beta)$, corresponding to $\alpha=\beta$ and $\alpha=0$, respectively.
Let $q:=\Pr(Y=1)$.  We say that perfect privacy yields a non-trivial utility if $\cP(Y|Z)>\cP(Y)$ for some $Z$ such that $\cP(X|Z)=\cP(X)$, or equivalently, if $\mathcalboondox{h}(p)>\max\{\bar{q}, q\}$. The following lemma determines $\mathcalboondox{h}(p)$ in the non-trivial case $\bar{\alpha}\bar{p}>\beta p$.
\begin{lemma}
\label{lemma_Perfect_Privacy_BIBO}
Let $X\sim\sBer(p)$ with $p\in[\frac{1}{2},1)$ and $P_{Y|X}=\mathsf{BIBO}(\alpha,\beta)$ with $\alpha,\beta\in[0,\frac{1}{2})$ such that $\bar{\alpha}\bar{p}>\beta p$. Then
\begin{equation*}\mathcalboondox{h}(p) = \begin{cases}1-\zeta q & \text{if}~ \alpha\bar{\alpha}\bar{p}^2<\beta\bar{\beta}p^2,\\q & \text{otherwise},\end{cases}
\end{equation*}
where $q=\alpha\bar{p}+\bar{\beta}p$ and 
\begin{equation}\label{Zeta_p}
  \zeta:=\frac{\bar{\alpha}\bar{p}-\beta p}{\bar{\beta}p-\alpha\bar{p}}.
\end{equation}
\end{lemma}

Notice that $1-\zeta q > \bar{q}$ if and only if $\zeta<1$, which occurs if and only if $p\in(\frac{1}{2},1)$. Also, it is straightforward to show that $1-\zeta q > q$ if and only if $\alpha\bar{\alpha}\bar{p}^2< \beta\bar{\beta}p^2$. In particular, we have the following necessary and sufficient condition for non-trivial utility under perfect privacy.

\begin{corollary}
Let $X\sim\sBer(p)$ with $p\in[\frac{1}{2},1)$ and $P_{Y|X}=\mathsf{BIBO}(\alpha,\beta)$ with $\alpha,\beta\in[0,\frac{1}{2})$ such that $\bar{\alpha}\bar{p}>\beta p$. Then $g^\infty(\sP, 0)>0$ if and only if $\alpha\bar{\alpha}\bar{p}^2<\beta\bar{\beta}p^2$ and $p\in(\frac{1}{2},1)$.
\end{corollary}
Remark that the condition $\alpha\bar{\alpha}\bar{p}^2< \beta\bar{\beta}p^2$ can be equivalently written  as $$P_{X|Y}(0|1)P_{X|Y}(0|0)<P_{X|Y}(1|0)P_{X|Y}(1|1).$$

The following theorem establishes the linear behavior of $\mathcalboondox{h}$ when $P_{Y|X}=\mathsf{BIBO}(\alpha, \beta)$.

\begin{theorem}\label{Theorem_Linearity_BIBO}
Let $X\sim \sBer(p)$ for $p\in [\frac{1}{2}, 1)$ and $P_{Y|X}=\mathsf{BIBO}(\alpha, \beta)$ with $\alpha,\beta\in[0,\frac{1}{2})$. If $\bar{\alpha}\bar{p}>\beta p$, then for any $\eps\in [p, \bar{\alpha}\bar{p}+\bar{\beta}p]$, we have the following:
\begin{itemize}
  \item If $\alpha\bar{\alpha}\bar{p}^2<\beta\bar{\beta}p^2$, then
      $$\mathcalboondox{h}(\eps)=1-\zeta(\eps) q,$$ where $q=\alpha\bar{p}+\bar{\beta}p$ and  \begin{equation}\label{zeta}
        \zeta(\eps):=\frac{\bar{\alpha}\bar{p}+\bar{\beta}p-\eps}{\bar{\beta} p-\alpha\bar{p}}.
      \end{equation}
      Furthermore, $\mathcalboondox{h}(\eps)$ is achieved by the Z-channel $\mathsf{Z}(\zeta(\eps))$ (as shown in Fig.\ \ref{fig:Optimal_Filter_BIBO}).
  \item If $\alpha\bar{\alpha}\bar{p}^2\geq\beta\bar{\beta}p^2$, then
            $$\mathcalboondox{h}(\eps)=1-\tilde{\zeta}(\eps) \bar{q},$$
where
$$\tilde{\zeta}(\eps):=\frac{\bar{\alpha}\bar{p}+\bar{\beta}p-\eps}{\bar{\alpha}\bar{p}-\beta p}.$$
Moreover, $\mathcalboondox{h}(\eps)$ is achieved by a reverse Z-channel with crossover probability $\tilde{\zeta}(\eps)$ (as shown in Fig.\ \ref{fig:Optimal_Filter_BIBO}).
\end{itemize}
\end{theorem}
\begin{proof}[Proof Sketch]
 Recall that $\eps\mapsto\mathcalboondox{h}(\eps)$ is concave, and thus its graph lies above the segment connecting $(p,\mathcalboondox{h}(p))$ to $(\cP(X|Y),1)$. In particular,
\begin{equation*}
  \mathcalboondox{h}(\eps) \geq \mathcalboondox{h}(p) + (\eps-p)\left[\frac{1-\mathcalboondox{h}(p)}{\cP(X|Y)-p}\right].
\end{equation*}
By Lemma \ref{lemma_Perfect_Privacy_BIBO}, the above inequality becomes
\begin{eqnarray}
  \mathcalboondox{h}(\eps) &\geq& \mathcalboondox{h}(p)+\frac{q(\eps-p)}{\bar{\beta}p-\alpha\bar{p}} 1_{\{\alpha\bar{\alpha}\bar{p}^2<\beta\bar{\beta}p^2\}} \nonumber\\
  &&\quad~~~~ +\frac{\bar{q}(\eps-p)}{\bar{\alpha}\bar{p}-\beta p}1_{\{\alpha\bar{\alpha}\bar{p}^2\geq\beta\bar{\beta}p^2\}}.\label{eq:hLowerBound}
\end{eqnarray}
Since $\eps\mapsto\mathcalboondox{h}(\eps)$ is piecewise linear, its right derivative exists at $\eps=\cP(X|Y)$.
Using the geometric properties of $\H$ used to prove Theorem~\ref{Thm:PiecewiseLinearity}, we can show that
\begin{eqnarray*}
  \mathcalboondox{h}'(\cP(X|Y))&=& \frac{q}{\bar{\beta}p-\alpha\bar{p}} 1_{\{\alpha\bar{\alpha}\bar{p}^2<\beta\bar{\beta}p^2\}}\\ &&+\frac{\bar{q}}{\bar{\alpha}\bar{p}-\beta p}1_{\{\alpha\bar{\alpha}\bar{p}^2\geq\beta\bar{\beta}p^2\}},
\end{eqnarray*}
which is equal to the slope of the chord connecting  $(p,\mathcalboondox{h}(p))$ to $(\cP(X|Y),1)$ described in \eqref{eq:hLowerBound}. The concavity of $\mathcalboondox{h}(\cdot)$ thus implies that the inequality \eqref{eq:hLowerBound} is indeed equality.
\end{proof}
\begin{figure}[t]
\subfigure[$\alpha\bar{\alpha}\bar{p}^2<\beta\bar{\beta}p^2$]{
\begin{tikzpicture}[scale=0.71]
\node (a) [circle] at (0,0) {$1~$};
\node (b) [circle] at (0,1.4) {$0~~$};
\node (c) [circle] at (2,0) {$~1$};
\node (d) [circle] at (2,1.4) {$~0$};
\node (ee1) [circle] at (2.1,0.7) {\textcolor{blue}{$Y$}};
\node (ee2) [circle] at (-0.1,0.7) {\textcolor{blue}{$X$}};
\node (ee3) [circle] at (4.35,0.7) {\textcolor{blue}{$Z$}};
\draw[->] (0.15,0) -- (1.85,0) node[pos=.5,sloped,below] {};
\draw[->] (0.15,0) -- (1.85,1.4) node[pos=.35,sloped,below] {\footnotesize{$\beta$}};
\draw[->] (0.15,1.4) -- (1.85,0) node[pos=.35,sloped,above] {\footnotesize{$\alpha$}};
\draw[->] (0.15,1.4) -- (1.85,1.4) node[pos=.5,sloped,above] {};
\node (e) [circle] at (2.3,0) {};
\node (f) [circle] at (2.3,1.4) {};
\node (g) [circle] at (4.27,0) {$~1$};
\node (h) [circle] at (4.27,1.4) {$~0$};
\draw[->] (2.3,0) -- (4,0) node[pos=.5,sloped,below] {\footnotesize{1-$\zeta(\eps)$}};
\draw[->] (2.3,0) -- (4, 1.4) node[pos=.5,sloped,above] {\footnotesize{$\zeta(\eps)$}};
\draw[->] (2.3,1.4) -- (4,1.4) node[pos=.5,sloped,above] {};
\end{tikzpicture}}
  \subfigure[$\alpha\bar{\alpha}\bar{p}^2\geq\beta\bar{\beta}p^2$]{
\begin{tikzpicture}[scale=0.71]
\vspace{-0.1cm}
        \node (a) [circle] at (0,0) {$1~$};
\node (b) [circle] at (0,1.4) {$0~~$};
\node (c) [circle] at (2,0) {$~1$};
\node (d) [circle] at (2,1.4) {$~0$};
\node (ee1) [circle] at (2.1,0.7) {\textcolor{blue}{$Y$}};
\node (ee2) [circle] at (-0.1,0.7) {\textcolor{blue}{$X$}};
\node (ee3) [circle] at (4.35,0.7) {\textcolor{blue}{$Z$}};
\draw[->] (0.15,0) -- (1.85,0) node[pos=.5,sloped,below] {};
\draw[->] (0.15,0) -- (1.85,1.4) node[pos=.35,sloped,below] {\footnotesize{$\beta$}};
\draw[->] (0.15,1.4) -- (1.85,0) node[pos=.35,sloped,above] {\footnotesize{$\alpha$}};
\draw[->] (0.15,1.4) -- (1.85,1.4) node[pos=.5,sloped,above] {};
\node (e) [circle] at (2.3,0) {};
\node (f) [circle] at (2.3,1.4) {};
\node (g) [circle] at (4.27,0) {$~1$};
\node (h) [circle] at (4.27,1.4) {$~0$};
\draw[->] (2.3,0) -- (4,0) node[pos=.5,sloped,below] {};
\draw[->] (2.3,1.4) -- (4,1.4) node[pos=.5,sloped,above] {\footnotesize{$1-\tilde{\zeta}(\eps)$}};
\draw[->] (2.3,1.4) -- (4, 0) node[pos=.5,sloped,above]{\footnotesize{$\tilde{\zeta}(\eps)$}}; 
\end{tikzpicture}}
  \caption{The optimal privacy filters for $P_{Y|X}=\mathsf{BIBO(\alpha, \beta)}$.}
  \label{fig:Optimal_Filter_BIBO}
\end{figure}
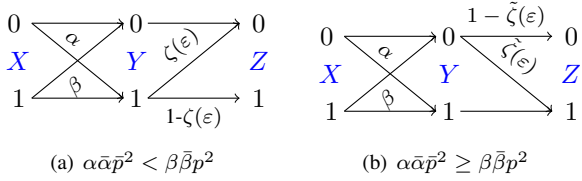
Under the hypotheses of the previous theorem, for every $\eps\in[\cP(X),\cP(X|Y)]$ there exists a Z-channel that achieves $\mathcalboondox{h}(\eps)$. It can be shown that  Z-channel is the \emph{only} binary filter with this property. It is also worth mentioning that even if $P_{Y|X}$ is symmetric (i.e., $\alpha=\beta$), the optimal filter cannot be symmetric, unless $X$ is uniform, in which case $\mathsf{BSC}(0.5\zeta(\eps))$ is also optimal.
\section{I.I.D. Binary Symmetric Vector Case}
We next study privacy aware guessing for a pair of binary random \emph{vectors} $(X^n, Y^n)$ with  $X^n, Y^n\in\{0, 1\}^n$. Recall that in this case it is sufficient to consider auxiliary random variables having supports of cardinality $2^n+1$. However, this condition may be practically inconvenient. Moreover,  in the scalar binary case examined in the last section we observed that a binary $Z$ was sufficient to achieve $\mathcalboondox{h}(\eps)$. Hence, it is natural to  require the privacy filters to produce also binary random vectors, i.e., $Z^n\in\{0,1\}^n$, which leads to the following definition. 
Recall that the data processing inequality implies that $\cP(X^n)\leq\cP(X^n|Z^n)\leq \cP(X^n|Y^n)$ and hence we can assume $\cP(X^n)\leq\eps^n\leq \cP(X^n|Y^n)$.
\begin{definition}\label{Def_h_n}
For a given pair of binary random vectors $(X^n, Y^n)$, we define  $\underbar{h}_n(\eps)$ for $\eps\in[\cP^{1/n}(X^n),\cP^{1/n}(X^n|Y^n)]$, as
\begin{equation}
\underbar{h}_n(\eps) \coloneqq \max~\cP^{1/n}(Y^n | Z^n), \end{equation}
where the maximum is taken over all (not necessarily memoryless) channels $P_{Z^n|Y^n}$ such that $Z^n\in\{0,1\}^n$, $X^n\markov Y^n\markov Z^n$, and $\cP(X^n|Z^n) \leq \eps^n$.
\end{definition}
Note that this definition does not make any assumption about the privacy filters $P_{Z^n|Y^n}$ except that $Z^n\in\{0,1\}^n$.
From an implementation point of view, the simplest privacy filter is a memoryless one such that  $Z_k$ is a noisy version of $Y_k$ for $k=1,\dots,n$. More precisely, we are interested in a \textit{single} $\mathsf{BIBO}$ channel $P_{Z|Y}$ which, given $Y_k$, generates $Z_k$ according to  $$P_{Z^n|Y^n}(z^n|y^n)=\prod_{k=1}^nP_{Z|Y}(z_k|y_k).$$  
Now, let $h_n^\mathsf{i}(\eps)$ be defined as $\max \cP^{1/n}(Y^n|Z^n)$, where the maximum is taken over  such memoryless privacy filters satisfying $\cP(X^n|Z^n)\leq \eps^n$.
Let $\oplus$ denote mod 2 addition. In what follows, we study $\underbar{\emph{h}}_n$ and $h_n^\mathsf{i}$ for the following setup:
\begin{itemize}
	\item[$\textnormal{a}$)] $\repdc{X}{1}{n}$ are i.i.d. $\sBer(p)$ random variables with $p\geq \frac{1}{2}$,
\item[b)] $Y_k=X_k\oplus V_k$ for $k=1, \dots, n$, where $\repdc{V}{1}{n}$ are i.i.d. $\sBer(\alpha)$ random variables independent of $X^n$, such that $\alpha< \frac{1}{2}$.
\end{itemize}
\begin{figure}[t]
	\centering
	\begin{tikzpicture}[scale=0.7]
	\node (a) [circle] at (0,0) {$11~~$};
	\node (b) [circle] at (0,1.8) {$01~~$};
	\node (c) [circle] at (2,0) {$~~~~~~~11~~~~~~$};
	\node (d) [circle] at (2,1.8) {$~~~~~~~01~~~~~~$};
	\node (a1) [circle] at (0,0.9) {$10~~$};
	\node (b1) [circle] at (2,0.9) {$~~~~~~~10~~~~~~$};
	\node (c1) [circle] at (2,2.7) {$~~~~~~~00~~~~~~$};
	\node (d1) [circle] at (0,2.7) {$00~~$};
	\draw[->] (0.15,0) -- (1.83,0) node[pos=.5,sloped,below] {};
	\draw[->] (0.15,0) -- (1.83,0.8) node[pos=.5,sloped,below] {};
	\draw[->] (0.15,0) -- (1.83,1.7) node[pos=.5,sloped,below] {};
	\draw[->] (0.15,0) -- (1.83,2.6) node[pos=.5,sloped,below] {};
	\draw[->] (0.15,1.7) -- (1.83,0) node[pos=.5,sloped,below] {};
	\draw[->] (0.15,0.8) -- (1.83,0.8) node[pos=.5,sloped,below] {};
	\draw[->] (0.15,1.7) -- (1.83,1.7) node[pos=.5,sloped,below] {};
	\draw[->] (0.15,2.6) -- (1.83,2.6) node[pos=.5,sloped,below] {};
	\draw[->] (0.15,2.6) -- (1.83,0) node[pos=.5,sloped,below] {};
	\draw[->] (0.15,1.7) -- (1.83,0.8) node[pos=.5,sloped,below] {};
	\draw[->] (0.15,0.8) -- (1.83,1.7) node[pos=.5,sloped,below] {};
	\draw[->] (0.15,0.8) -- (1.83,2.6) node[pos=.5,sloped,below] {};
	\draw[->] (0.15,0.8) -- (1.83,0) node[pos=.5,sloped,below] {};
	\draw[->] (0.15,2.6) -- (1.83,0.8) node[pos=.5,sloped,below] {};
	\draw[->] (0.15,2.6) -- (1.83,1.7) node[pos=.5,sloped,below] {};
	\draw[->] (0.15,1.7) -- (1.83,2.6) node[pos=.5,sloped,below] {};
	\node (g) [circle] at (4.47,0) {$~~~11~~~$};
	\node (h) [circle] at (4.47,0.9) {$~~~10~~$};
	\node (g1) [circle] at (4.47,1.8) {$~~~01~~$};
	\node (h1) [circle] at (4.47,2.6) {$~~~00~~$};
	\draw[->] (2.35,0) -- (4.2,2.6) node[pos=.65,sloped,below] {};
	\draw[->] (2.35,0) -- (4.2, 0) node[pos=.55,sloped,below] {\footnotesize{1-$\zeta_2(\eps)$}};
	\draw[->] (2.35,0.8) -- (4.2,0.8) node[pos=.5,sloped,above] {};
	\draw[->] (2.35,1.7) -- (4.2,1.7) node[pos=.5,sloped,above] {};
	\draw[->] (2.35,2.6) -- (4.2,2.6) node[pos=.5,sloped,above] {};
	\end{tikzpicture}
	\caption{The optimal privacy filter for $\underbar{\emph{h}}_2(\eps)$ for $\eps\in[\eps_\mathsf{L}, \bar{\alpha})$, where $\zeta_2(\eps)$ is defined in \eqref{Zeta_n_vector}.  }\label{fig:Optimal_Vector}
\end{figure}
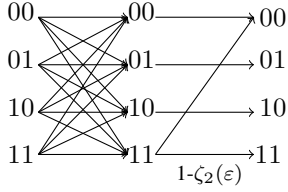
We first determine $h_n^\mathsf{i}(\eps)$ for this model and show that (as expected) $h_n^\mathsf{i}(\eps)$ is independent of $n$. According to this model, $\cP(X^n)=p^n$ and $\cP(X^n|Y^n)=\bar{\alpha}^n$, and thus $p\leq\eps\leq \bar{\alpha}$.
\begin{proposition}\label{Propo_h_iid}
If $(X^n, Y^n)$ satisfies a) and b) with $p\in [\frac{1}{2}, 1)$ and $\alpha\in [0, \frac{1}{2})$ such that $\bar{\alpha}>p$, then
\begin{equation*}h_n^\mathsf{i}(\eps) = \mathcalboondox{h}(\eps)=1-\zeta(\eps)q,\end{equation*}
for all $\eps\in[p,\bar{\alpha}]$, where $\zeta(\eps)$ is given in \eqref{zeta} and $q=\alpha\bar{p}+\bar{\alpha}p$.
\end{proposition}
Note that the proposition reduces to Theorem~\ref{Theorem_Linearity_BIBO} for $n=1$. However, for $n\geq 2$, we have $h_n^\mathsf{i}(\eps)<\underbar{\emph{h}}_n(\eps)\leq \mathcalboondox{h}(P_{X^nY^n}, \eps),$ as implied by the following theorem. A channel $\mathsf{W}$ is said to be  a $2^n$-ary Z-channel, denoted by $\mathsf{Z}_n(\gamma)$, if the input and output alphabets are $\{0, 1\}^n$ and $\mathsf{W}(a|a)=1$ for $a\neq \bf{1}$, $\mathsf{W}(\bf{0}|\bf{1})=\gamma$, and $\mathsf{W}(\bf{1}|\bf{1})=\bar{\gamma}$, where ${\bf 0}=(0, 0, \dots, 0)$ and ${\bf 1}=(1, 1, \dots, 1)$.
\begin{theorem}
\label{Thm:Difference}
Assume that $(X^n,Y^n)$ satisfies a) and b) with $p\in[\frac{1}{2},1)$ and $\alpha\in [0,\frac{1}{2})$ such that $\bar{\alpha}>p$. Then, there exists $p\leq\eps_\mathsf{L}<\bar{\alpha}$  such that
$$\underbar{h}^n_n(\eps)= 1-\zeta_n(\eps)q^n,$$
for $\eps\in [\eps_\mathsf{L}, \bar{\alpha}]$, where $q=\alpha\bar{p}+\bar{\alpha}p$ and 
\begin{equation}\label{Zeta_n_vector}
  \zeta_n(\eps)\coloneqq \frac{\bar{\alpha}^n-\eps^n}{(\bar{\alpha}p)^n-(\alpha\bar{p})^n}.
\end{equation}
Moreover, the channel $\mathsf{Z}_n(\zeta_n(\eps))$ achieves $\underbar{h}_n(\eps)$ in this interval (see Fig.~\ref{fig:Optimal_Vector} for the case $n=2$).
\end{theorem}

The memoryless privacy filter assumed in $h_n^\mathsf{i}(\eps)$ is simple to implement. However, it is clear from Theorem~\ref{Thm:Difference} that this simple filter is not optimal even when $(X^n, Y^n)$ is i.i.d.\ since $\underbar{\emph{h}}_n(\eps)$ is a function of $n$, while $h_n^\mathsf{i}(\eps)$ is not. The following corollary bounds the loss resulting from using a simple memoryless filter instead of an optimal one for $\eps\in [\eps_\mathsf{L}, \bar{\alpha}]$. Clearly, for $n=1$, there is no gap as $\underbar{\emph{h}}_1(\eps)=h^\mathsf{i}_1(\eps)$.
\begin{corollary}\label{Thm:Diference_Gap}
Let $(X^n,Y^n)$ satisfy a) and b) with $p\in[\frac{1}{2},1)$ 
and $\alpha\in [0,\frac{1}{2})$ such that $\bar{\alpha}>p$. If $p>\frac{1}{2}$ and $\alpha>0$, then for $\eps\in [\eps_\mathsf{L}, \bar{\alpha}]$ and sufficiently large $n$
\begin{equation}\label{LB_Gap}
\underbar{h}_n(\eps)-h^\mathsf{i}_n(\eps)\geq (\bar{\alpha}-\eps)[\Phi(1)-\Phi(n)],
\end{equation}
where $$\Phi(n)\coloneqq \frac{q^{n}\bar{\alpha}^{n-1}}{(\bar{\alpha}p)^n-(\alpha\bar{p})^n}.$$
If $p=\frac{1}{2}$, then
\begin{equation}\label{UB_Gap}
 h^\mathsf{i}_n(\eps) \leq \underbar{h}_n(\eps)\leq h^\mathsf{i}_n(\eps) + \frac{\alpha}{2\bar{\alpha}}, 
\end{equation}
for every $n\geq 1$ and $\eps\in [\eps_\mathsf{L}, \bar{\alpha}]$.
\end{corollary}
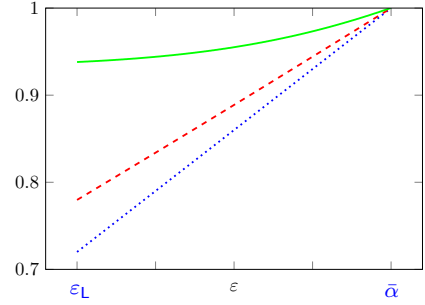
\begin{figure}
\centering
\begin{tikzpicture}[thick, scale=0.7]
\begin{axis}[axis on top, xticklabels={$\eps_{\mathsf{L}}$, $\bar{\alpha}$},ymin=0.7,ymax=1,
  samples=50, y=16.5cm, x label style={at={(axis description cs:0.5,0)},rotate=0,anchor=south},
    x=29.5cm, xlabel= $\eps$,
  no markers]
\addplot+[domain=0.6:0.8][name path=A, no marks,blue,dotted, line width=1pt,samples=600] {1.4*x-0.12};
\addplot+[domain=0.6:0.8][name path=B, no marks,red, dashed, line width=1pt,samples=600] {sqrt(1.4*x^2+0.104)};
\addplot[domain=0.6:0.8][name path=B2, no marks,green,line width=1pt,samples=600] {4.67162*x^(10)+0.498388)^(0.1)};
\end{axis}
\node[below,blue] at (0.65, -0.1) {\footnotesize{$\eps_{\mathsf{L}}$}};
\node[below,blue] at (6.5, -0.1) {\footnotesize{$\bar{\alpha}$}};
\end{tikzpicture}
\caption{The graphs of $\underbar{\emph{h}}_{10}$ (solid curve), $\underbar{\emph{h}}_{2}$ (dashed curve), and $\emph{h}^{\mathsf{i}}$ (dotted line) given in Theorem~\ref{Thm:Difference} and Proposition~\ref{Propo_h_iid} for i.i.d.\ $(X^n, Y^n)$ with $X\sim\sBer(0.6)$ and $P_{Y|X}=\mathsf{BSC}(0.2)$.}\label{fig:1}
\end{figure}  
Since $\Phi(n)\downarrow 0$ as $n\to \infty$, \eqref{LB_Gap} implies that, as expected, the gap between the performance of the optimal privacy filter and the optimal memoryless privacy filter increases as $n$ increases. This observation is numerically illustrated in Fig.~\ref{fig:1}, where $\underbar{\emph{h}}_n(\eps)$ is plotted as a function of $\eps$ for $n=2$ and $n=10$.
Moreover, \eqref{UB_Gap} implies that when $p=\frac{1}{2}$ and $\alpha$ is small, then $\underbar{\emph{h}}_n(\eps)$ can be approximated by $h^\mathsf{i}_n(\eps)$. Thus, we can approximate the optimal filter $\mathsf{Z}_n(\zeta_n(\eps))$ with a simple memoryless filter given by $Z_k=Y_k\oplus W_k$, where $W_1, \dots, W_n$ are i.i.d.\ $\sBer(0.5\zeta(\eps))$ random variables that are independent of $(X^n, Y^n)$.

\bibliographystyle{IEEEtran}
\bibliography{bibliography}

\end{document}